\documentclass[a4paper,11pt,fleqn]{article}

\usepackage{amsmath, amssymb, amsthm}
\usepackage{mathtools, thm-restate}
\usepackage{mathrsfs}
\usepackage[margin=1in]{geometry}
\usepackage{xcolor}
\usepackage{url}
\usepackage{comment}
\usepackage{tikz}
\usepackage{enumitem}
\usepackage{array}
\usepackage{longtable}
\usepackage{diagbox}
\usepackage{anyfontsize}
\usepackage{xfrac}
\usepackage[justification=centering]{caption}
\usepackage[hypertexnames=false,final]{hyperref}
\usepackage{algorithm, algpseudocode}
\usepackage[capitalize,sort]{cleveref}
\usetikzlibrary{decorations.pathreplacing,arrows.meta,calc}

\usepackage{amsmath}

\let\epsilon\varepsilon

\newcommand*{\thmdep}[2]{}
\newcommand*{\thmdepcref}[2]{\cref{#1}}

\newcommand*{\Th}{^{\textrm{th}}}

\let\eps\epsilon

\newcommand*{\OPT}{\mathrm{OPT}}
\newcommand*{\defeq}{:=}
\newcommand*{\pfrac}[2]{\left(\frac{#1}{#2}\right)}

\newcommand*{\ceil}[1]{\left\lceil #1 \right\rceil}

\newcommand*{\abs}[1]{\left\lvert #1 \right\rvert}
\newcommand*{\smallabs}[1]{\lvert #1 \rvert}
\newcommand*{\norm}[1]{\left\lVert #1 \right\rVert}

\DeclareMathOperator*{\E}{E}

\DeclareMathOperator*{\argmin}{argmin}
\DeclareMathOperator*{\argmax}{argmax}
\DeclareMathOperator*{\support}{support}
\DeclareMathOperator*{\opt}{opt}

\DeclareMathOperator{\poly}{poly}

\usepackage{iftex}
\usepackage{url}
\usepackage{mathtools}
\usepackage{comment}

\hypersetup{
    colorlinks,
    linkcolor={red!50!black},
    citecolor={red!50!black},
    urlcolor={blue!50!black}
}

\makeatletter
\g@addto@macro{\UrlBreaks}{%
\do\/%
\do\a\do\b\do\c\do\d\do\e\do\f\do\g\do\h\do\i\do\j\do\k\do\l\do\m%
\do\n\do\o\do\p\do\q\do\r\do\s\do\t\do\u\do\v\do\w\do\x\do\y\do\z%
\do\A\do\B\do\C\do\D\do\E\do\F\do\G\do\H\do\I\do\J\do\K\do\L\do\M%
\do\N\do\O\do\P\do\Q\do\R\do\S\do\T\do\U\do\V\do\W\do\X\do\Y\do\Z%
\do\0\do\1\do\2\do\3\do\4\do\5\do\6\do\7\do\8\do\9%
}
\makeatother

\renewcommand*{\defeq}{\coloneqq}
\newcolumntype{L}{>{$\displaystyle}l<{$}}
\algnewcommand{\LineComment}[1]{\State \textcolor{gray}{// #1}}

\newcommand*{\restate}[2]{\textbf{#1.} \emph{#2}}

\newcommand*{\setParSpacing}{
    \setlength{\parindent}{0pt}
    \setlength{\parskip}{0.6em}
}
\newcommand*{\acknowledgements}[1]{\paragraph{Acknowledgements.} #1}

\newtheorem{theorem}{Theorem}
\newtheorem{definition}{Definition}

\newtheorem{corollary}{Corollary}[theorem]
\newtheorem{lemma}[theorem]{Lemma}
\newtheorem{claim}[theorem]{Claim}

\crefname{claim}{Claim}{Claims}
\crefname{property}{Property}{Properties}
\crefname{observation}{Observation}{Observations}
\crefname{transformation}{Transformation}{Transformations}

\newcommand*{\config}{configuration}
\newcommand*{\Config}{Configuration}

\newcommand*{\Ccal}{\mathcal{C}}

\newcommand*{\Otild}{\widetilde{O}}
\newcommand*{\xhat}{\widehat{x}}
\newcommand*{\xtild}{\widetilde{x}}
\newcommand*{\yhat}{\widehat{y}}
\newcommand*{\lamhat}{\widehat{\lambda}}
\newcommand*{\lamtild}{\widetilde{\lambda}}
\newcommand*{\epsSigma}{\eps_{\sigma}}
\newcommand*{\epsAlpha}{\eps_{\alpha}}
\newcommand*{\vecone}{\mathbf{1}}
\newcommand*{\prodOrHyp}{\hyperref[defn:product-oracle]{product oracle}}
\newcommand*{\columnOrHyp}{\hyperref[defn:column-oracle]{column oracle}}
\newcommand*{\costOrHyp}{\hyperref[defn:cost-oracle]{cost oracle}}
\newcommand*{\indexOrHyp}{\hyperref[defn:index-find]{index-finding oracle}}
\newcommand*{\pointOrHyp}{\hyperref[defn:point-find]{point-finding oracle}}

\DeclareMathOperator{\fcov}{fcov}
\DeclareMathOperator{\fcovHyp}{\hyperref[defn:frac-cov]{\fcov}}
\DeclareMathOperator{\ofcov}{ofcov}
\DeclareMathOperator{\dfcov}{dfcov}
\DeclareMathOperator{\width}{width}
\DeclareMathOperator{\widthHyp}{\hyperref[defn:width]{\width}}
\DeclareMathOperator{\covLP}{covLP}
\DeclareMathOperator{\CondI}{\mathcal{C}_1}
\DeclareMathOperator{\CondIHyp}{\hyperref[eqn:condI]{\CondI}}
\DeclareMathOperator{\CondII}{\mathcal{C}_2}
\DeclareMathOperator{\CondIIHyp}{\hyperref[eqn:condII]{\CondII}}

\DeclareMathOperator{\getSeed}{\mathtt{get-seed}}
\DeclareMathOperator{\getSeedHyp}{\hyperref[algo:get-seed]{\getSeed}}
\DeclareMathOperator{\improveCover}{\mathtt{improve-cover}}
\DeclareMathOperator{\improveCoverHyp}{\hyperref[algo:improve-cover]{\improveCover}}
\DeclareMathOperator{\fracCover}{\mathtt{frac-cover}}
\DeclareMathOperator{\fracCoverHyp}{\hyperref[algo:frac-cover]{\fracCover}}
\DeclareMathOperator{\fracCovII}{\mathtt{frac-cover-2}}
\DeclareMathOperator{\indexFind}{\mathtt{index-find}}
\DeclareMathOperator{\indexFindHyp}{\hyperref[defn:index-find]{\indexFind}}
\DeclareMathOperator{\pointFind}{\mathtt{point-find}}
\DeclareMathOperator{\pointFindHyp}{\hyperref[defn:point-find]{\pointFind}}
\DeclareMathOperator{\covLPsolve}{\mathtt{covLP-solve}}
\DeclareMathOperator{\covLPsolveHyp}{\hyperref[algo:cov-lp-solve]{\covLPsolve}}

\title{An Approximation Algorithm for Covering Linear
Programs\texorpdfstring{\\}{ }and its Application to Bin-Packing}
\author{Eklavya Sharma\\
Department of Computer Science and Automation\\
Indian Institute of Science, Bengaluru.\\
\texttt{eklavyas@iisc.ac.in}}
\date{\empty}

\begin{document}

\maketitle

\begin{abstract}
We give an $\alpha(1+\eps)$-approximation algorithm for solving covering LPs, assuming
the presence of a $(1/\alpha)$-approximation algorithm for a certain optimization problem.
Our algorithm is based on a simple modification of the Plotkin-Shmoys-Tardos algorithm
\cite{plotkin1995fast}.
We then apply our algorithm to $\alpha(1+\eps)$-approximately solve the
\config{} LP for a large class of bin-packing problems,
assuming the presence of a $(1/\alpha)$-approximate algorithm for the
corresponding knapsack problem (KS).
Previous results give us a PTAS for the \config{} LP using a PTAS for KS.
Those results don't extend to the case where KS is poorly approximated.
Our algorithm, however, works even for polynomially-large $\alpha$.

\end{abstract}

\setParSpacing
\setlist[itemize]{noitemsep,nolistsep}

\acknowledgements{I want to thank my advisor, Prof.~Arindam Khan, for his valuable comments,
and K.V.N.~Sreenivas for helpful discussions.}

\section{Introduction}

Algorithms for solving linear programs (LPs) have been the cornerstone of operations research.
Linear programming also has applications in computer science;
Gr\"otschel, Lov\'asz and Schrijver \cite{gls-ellipsoid} give examples of
combinatorial optimization problems that can be solved using linear programming.
Linear programs are especially important in the area of approximation algorithms.
Many optimization problems can be expressed as integer programs.
\emph{Rounding-based} algorithms first solve the LP relaxation of these integer programs,
and then \emph{round} the relaxed solution to get an approximate solution to the original problem
\cite{det-lp-round-daa,rand-lp-round-daa,iterative-methods}.

We study a large and important class of linear programs, called \emph{covering linear programs}.
Our main result is an approximation algorithm, called $\covLPsolveHyp$,
for solving covering LPs.

\begin{definition}
\label{defn:cov-lp}
A linear program is called a covering LP iff it is of the form
\[ \min_{x \in \mathbb{R}^N} c^Tx \textrm{ where } Ax \ge b \textrm{ and } x \ge 0, \]
where $A \in \mathbb{R}_{\ge 0}^{m \times N}$ ($m$-by-$N$ matrix over non-negative reals),
$b \in \mathbb{R}^m_{>0}$ %
and $c \in \mathbb{R}^N_{> 0}$. %
Denote this covering LP by $\covLP(A, b, c)$.
\end{definition}

Our motivating application stems from the bin-packing problem.
There are multiple ways of representing bin-packing as an integer LP,
but probably the most useful of them is the \emph{\config{} LP}
(formally defined in \cref{sec:bin-packing}).
Rounding the \config{} LP was used in the first linear-time APTAS for bin-packing
by de la Vega and Lueker \cite{bp-aptas}.
It was later used by Karmarkar and Karp \cite{karmarkar-karp} to get an algorithm
for bin-packing that uses $\OPT + O(\log^2(\OPT))$ bins,
and by Hoberg and Rothvoss \cite{hoberg2017logarithmic}
for an algorithm that uses $\OPT + O(\log(\OPT))$ bins.
Bansal, Caprara and Sviridenko \cite{rna} devised the \emph{Round-and-Approx} (R\&A) framework
for solving variants of the bin-packing problem.
The R\&A framework requires an approximate solution to the \config{} LP.
They used the R\&A framework to get approximation algorithms
for vector bin-packing and 2-dimensional geometric bin-packing.
Improved algorithms were later devised for these bin-packing variants
\cite{bansal2014binpacking,BansalE016},
but those algorithms also use the R\&A framework.
Thus, solving the \config{} LP of (variants of) bin-packing is an important problem.
In \cref{sec:bin-packing}, we show how to approximately solve the \config{} LP
of a large class of bin-packing problems.

An implicit covering LP is one where $A$ and $c$ are not given to us explicitly.
Instead, we are given an input $I$, and $A$, $b$, $c$ are defined in terms of $I$.
The \config{} LP for bin-packing, for example, is defined implicitly.
Such an implicit definition is helpful when $N$, the number of columns in $A$,
is super-polynomial in the input size $|I|$.
We assume that $m$, the number of rows in $A$, is polynomial in $|I|$
and that $b$ has already been computed.
Since $A$ and $c$ are not given to us explicitly, we will assume the presence of certain oracles
that can help us indirectly get useful information about $A$ and $c$.
Our main result is an approximation algorithm $\covLPsolveHyp$
(described in \cref{sec:covLP-solve})
that solves $\covLP(A, b, c)$ in polynomial time using these oracles.

The main implication of our result is that for any $\eps > 0$, we can
$\alpha(1+\eps)$-approximately solve the \config{} LP of some variants of bin-packing,
using a $(1/\alpha)$-approximation algorithm for the corresponding knapsack problem.
Previous results give us a PTAS for the \config{} LP
using a PTAS for the corresponding knapsack problem
(see \cref{sec:prior-work} for details).
For many variants of knapsack, a PTAS is not known, so previous results cannot be applied.
Our algorithm, however, works even for polynomially-large $\alpha$.

\subsection{Formal Statement of Our Results}

\paragraph{Preliminaries:}
\begin{itemize}
\item For a non-negative integer $n$, let $[n] \defeq \{1, 2, \ldots, n\}$.
\item Let $\mathbb{R}_{\ge 0}$ be the set of non-negative real numbers.
Let $\mathbb{R}_{> 0} \defeq \mathbb{R}_{\ge 0} - \{0\}$.
\item For a vector $x$, $\support(x) \defeq \{j: x_j \neq 0\}$.
\item For a vector $x$, $x \ge 0$ means that every coordinate of $x$ is non-negative.
\item For a matrix $A$, $A[i, j]$ is the entry in the $i\Th$ row and $j\Th$ column of $A$.
\item Let $e_j$ be a vector whose $j\Th$ component is 1 and all other components are 0.
\item $\poly(n)$ is the set of functions of $n$ that are upper-bounded
by a polynomial in $n$.
\end{itemize}

\begin{definition}[Column oracle]
\label{defn:column-oracle}
The column oracle for $A \in \mathbb{R}_{\ge 0}^{m \times N}$
takes $j \in [N]$ as input and returns the $j\Th$ column of $A$.
\end{definition}
\begin{definition}[Cost oracle]
\label{defn:cost-oracle}
The cost oracle for $c \in \mathbb{R}_{> 0}^N$
takes $j \in [N]$ as input and returns $c_j$.
\end{definition}

\begin{definition}[Index-finding oracle]
\label{defn:index-find}
Let $A \in \mathbb{R}_{\ge 0}^{m \times N}$ and $c \in \mathbb{R}^N_{> 0}$
be implicitly defined in terms of input $I$.
For $j \in [N]$, define the function
$D_j: \mathbb{R}^m_{\ge 0} \mapsto \mathbb{R}_{\ge 0}$ as
\[ D_j(y) \defeq y^TA\left(\frac{e_j}{c_j}\right) = \frac{1}{c_j} \sum_{i=1}^m y_iA[i, j] \]
Then for $\eta \in (0, 1]$, an $\eta$-weak index-finding oracle for $I$,
denoted by $\indexFind$, is an algorithm that takes as input $y \in \mathbb{R}^m_{\ge 0}$
and returns $k \in [N]$ such that $D_k(y) \ge \eta \max_{j=1}^N D_j(y)$.
\end{definition}

The algorithm $\covLPsolve$ takes the following inputs:
\begin{itemize}
\item $I$: the input used to implicitly define $\covLP(A, b, c)$.
\item $q$: an upper-bound on $\opt(\covLP(A, b, c))$.
\item $\rho$: an upper-bound on
${\displaystyle q\max_{i=1}^m \max_{j=1}^N \frac{A[i,j]}{b_ic_j}}$.
\item $\eps, \eta \in (0, 1]$.
\end{itemize}
$\covLPsolve$ is also provided a column oracle for $A$, a cost oracle for $c$,
and an $\eta$-weak index-finding oracle.

\Cref{thm:cov-lp-solve,thm:covlps-time} below are our main results,
the proofs of which can be found in \cref{sec:covLP-solve}.

\newcommand*{\thmCovLPSolve}{%
Let $\covLP(A, b, c)$ be implicitly defined in terms of input $I$. Then
\\ $\covLPsolve(I, q, \rho, \eps, \eta)$ returns a
$(1+\eps+\eps^2)/\eta$-approximate solution to $\covLP(A, b, c)$.}
\begin{theorem}
\label{thm:cov-lp-solve}
\thmCovLPSolve
\end{theorem}

\newcommand*{\Uexpr}{m + \ceil{\ln\left(\frac{m}{\eta}\right)}
    \ceil{\frac{312m\rho(1+\eps)}{\eta\eps^3}\ln\left(\frac{12m}{\eps}\right)}}
\newcommand*{\thmCovLPSolveTime}{%
Let $\covLP(A, b, c)$ be implicitly defined in terms of input $I$,
where $A \in \mathbb{R}^{m \times N}_{\ge 0}$. Let
\begin{align*}
M &\defeq 3 + 2\lg\left(\frac{1}{\eps}+1\right)
    + \lg\left(\frac{1}{\eta}\right) + \lg\left(\frac{q}{\opt(\covLP(A, b, c))}\right)
\\ U &\defeq \Uexpr \in \Otild\left(\frac{m\rho}{\eta\eps^3}\right)
\end{align*}
Then all of the following hold for $\covLPsolve(I, q, \rho, \eps, \eta)$:
\begin{itemize}
\item $\covLPsolve$ makes at most $MU$ calls to the index-finding oracle,
at most $MU$ calls to the column oracle, and at most $MU$ calls to the cost oracle.
\item In $\covLPsolve$, the time taken by non-oracle operations is $O(MUm)$.
\item The solution $\xhat$ returned by $\covLPsolve$ has $|\support(\xhat)| \le U$.
\end{itemize}}
\begin{theorem}
\label{thm:covlps-time}
\thmCovLPSolveTime
\end{theorem}

\begin{corollary}
Let $r^* = \opt(\covLP(A, b, c))$.
Then to approximately solve $\covLP(A, b, c)$ in polynomial time using $\covLPsolve$,
we need a way to compute $\rho$ and $q$ in polynomial time,
$m, 1/\eps, 1/\eta, \log(q/r^*), \rho \in \poly(|I|)$,
and the oracles should run in $\poly(|I|)$ time.
\end{corollary}

\subsection{The Fractional Covering Problem}

A subsidiary contribution of this paper is an algorithm
for the \emph{fractional covering problem}.
We use that algorithm as a subroutine in $\covLPsolve$.

\begin{definition}[Fractional Covering Problem \cite{plotkin1995fast}]
\label{defn:frac-cov}
Let $A \in \mathbb{R}_{\ge 0}^{m \times N}$ be an $m$-by-$N$ matrix
and $b \in \mathbb{R}^m_{> 0}$ be an $m$-dimensional vector.
Let $P \subseteq \mathbb{R}^N$ be a convex polytope
such that $Ax \ge 0$ for all $x \in P$.
The fractional covering problem on input $(A, b, P)$, denoted as $\fcov(A, b, P)$,
requires us to do one of the following:
\begin{itemize}
\item output a feasible solution $x$,
i.e. $x \in P$ such that $Ax \ge b$.
\item claim that $\fcov(A, b, P)$ is unsatisfiable,
i.e. $Ax \ge b$ is not satisfied by any $x \in P$.
\end{itemize}
For $\alpha \in (0, 1]$, an algorithm is said to $\alpha$-weakly solve
$\fcov(A, b, P)$ iff it does one of the following:
\begin{itemize}
\item output an $\alpha$-approximate solution $x$,
i.e. $x \in P$ such that $Ax \ge \alpha b$.
\item claim that $\fcov(A, b, P)$ is unsatisfiable.
\end{itemize}
\end{definition}

An implicit fractional covering problem is one where $A$, $b$ and $P$
are not given to us explicitly. Instead, we are given an input $I$,
and $A$, $b$, $P$ are defined in terms of $I$.

We give an algorithm $\fracCoverHyp$ (described in \cref{sec:frac-cover})
that weakly solves $\fcov(A, b, P)$ in polynomial time using certain oracles.
$\fracCover$ is obtained by modifying the algorithm of Plotkin, Shmoys and Tardos
\cite{plotkin1995fast} for the fractional covering problem.
Moreover, we use $\fracCover$ as a subroutine in $\covLPsolve$.

\begin{definition}[Product oracle]
\label{defn:product-oracle}
Let $\fcov(A, b, P)$ be a fractional covering problem instance,
where $A, b, P$ are defined implicitly in terms of $I$.
The product oracle for $I$ takes $x \in P$ as input and returns $Ax$.
\end{definition}

\begin{definition}[Point-finding oracle]
\label{defn:point-find}
Let $\fcov(A, b, P)$ be a fractional covering problem instance,
where $A, b, P$ are defined implicitly in terms of $I$.
For $\eta \in (0, 1]$, an $\eta$-weak point-finding oracle for $I$,
denoted by $\pointFind$, is an algorithm that takes as input
$y \in \mathbb{R}^m_{\ge 0}$ and returns $\xhat \in P$ such that
$y^TA\xhat \ge \eta \max_{x \in P} y^TAx$.
\end{definition}

\begin{definition}
\label{defn:width}
\[ \width(A, b, P) \defeq \max_{x \in P} \max_{i=1}^m \frac{(Ax)_i}{b_i} \]
\end{definition}

The algorithm $\fracCover$ takes the following inputs:
\begin{itemize}
\item $I$: the input used to implicitly define $\fcov(A, b, P)$.
\item $\rho$: an upper-bound on $\width(A, b, P)$.
\item $\eps, \eta \in (0, 1]$.
\end{itemize}
$\fracCover$ is also provided a product oracle for $A$
and an $\eta$-weak point-finding oracle.

\Cref{thm:frac-cover,thm:frac-cover-time} below are the main results for this problem,
the proofs of which can be found in \cref{sec:frac-cover}.

\newcommand*{\thmFracCover}{%
Let $\fcov(A, b, P)$ be a fractional covering problem instance
where $A, b, P$ are implicitly defined in terms of input $I$.
Then $\fracCover(I, \rho, \eps, \eta)$ will $\eta/(1+\eps)$-weakly solve $\fcov(A, b, P)$,
i.e., if it returns \texttt{null}, then $\fcov(A, b, P)$ is unsatisfiable,
and if it returns a vector $x$, then $x \in P$ and $Ax \ge (\eta/(1+\eps))b$.}
\begin{theorem}
\label{thm:frac-cover}
\thmFracCover
\end{theorem}

\newcommand*{\thmFracCoverTime}{%
Let $\fcov(A, b, P)$ be implicitly defined in terms of input $I$,
where $A \in \mathbb{R}^{m \times N}_{\ge 0}$.
Let $\tau$ be an upper-bound on the support of the output of $\pointFind$.
Suppose $\fracCover(I, \rho, \eps, \eta)$ calls the point-finding oracle $T$ times. Then
\[ T \le U \defeq \Uexpr \in \Otild\left(\frac{m\rho}{\eta\eps^3}\right) \]
Additionally,
\begin{itemize}
\item $\fracCover$ makes at most $T$ calls to the product oracle.
For every input $x$ to the product oracle, $|\support(x)| \le \tau$.
\item The running time of $\fracCover$, excluding the time taken by oracles,
is $O(T(m + \tau))$.
\item The solution $\xhat$ returned by $\fracCover$ has $|\support(\xhat)| \le T\tau$.
\end{itemize}}
\begin{theorem}
\label{thm:frac-cover-time}
\thmFracCoverTime
\end{theorem}

\subsection{Organization of This Paper}

\Cref{sec:bin-packing} defines the bin-packing problem and its corresponding \config{} LP,
and shows how to use $\covLPsolve$ to approximately solve the \config{} LP in polynomial time.
\Cref{sec:prior-work} describes some well-known algorithms for solving LPs
and compares them to our algorithm.
\Cref{sec:frac-cover} describes the $\fracCover$ algorithm and proves
\cref{thm:frac-cover,thm:frac-cover-time}.
\Cref{sec:covLP-solve} describes the $\covLPsolve$ algorithm and proves
\cref{thm:cov-lp-solve,thm:covlps-time}.
\Cref{sec:future-work} describes avenues for further improvement.

\section{The Bin-Packing Problem}
\label{sec:bin-packing}

In this section, we will define the bin-packing problem and its corresponding \config{} LP.
Then we will see how to apply $\covLPsolve$ to approximately solve the \config{} LP
in polynomial time.

In the classic bin-packing problem (classic-BP), we are given a set $I$ of $n$ items.
Each item $i$ has a size $s_i \in (0, 1]$.
We want to partition $I$ such that the sum of sizes of items in each partition is at most 1.
Each partition is called a bin, and we want to minimize the number of bins.
We want an algorithm for this problem whose worst-case running time is polynomial in $n$.
See the survey by Coffman et al. on approximation algorithms for bin-packing \cite{coffman2013bin}.

In the classic knapsack problem (classic-KS), we are given a set $I$ of $n$ items.
Each item $i$ has a size $s_i \in (0, 1]$ and a profit $p_i$ associated with it.
We want to select a subset $J \subseteq I$ of items such that
$p(J) \defeq \sum_{i \in J} p_i$ is maximized.

There are many variants of the classic bin-packing problem.
In the 2D geometric bin-packing problem (2GBP) \cite{CKPT17},
we are given a set $I$ of $n$ axis-parallel rectangular items,
and we have to place the items into the minimum number of
rectangular bins without rotating the items, such that no two items overlap.
In the vector bin-packing problem (VBP), we are given a set $I$ of $n$ vectors over
$\mathbb{R}^d_{\ge 0}$ that we have to pack into the minimum number of bins such that
in each bin, the maximum coordinate of the sum of vectors is at most 1.
We can similarly define 2D geometric knapsack (2GKS) and vector knapsack (VKS).
Note that in classic-BP and VBP, we only need to partition the items into bins,
whereas in 2GBP, we also need to decide the position of the items into the bins.

\subsection{Abstract Bin-Packing and the \Config{} LP}

We will now state the bin-packing problem and the knapsack problem abstractly,
so that our results hold for a large class of their variants.
Let $I$ be a set of $n$ items.
A \config{} is a packing of some items from $I$ into a bin.
Let $\Ccal$ be the set of all possible \config{}s of $I$.
In the abstract bin-packing problem (BP),
we have to pack the items into the minimum number of bins, such that
the packing in each bin is according to some \config{} in $\Ccal$.
The abstract knapsack problem (KS) requires us to choose
the max-profit \config{} where each item has an associated profit.
Note that we can get different variants of BP and KS by defining $\Ccal$ appropriately.
For example, when $\Ccal = \{X: X \subseteq I \textrm{ and } \sum_{i \in X} s_i \le 1 \}$,
we get classic-BP and classic-KS.

We will now formulate BP as an integer linear program.
Let there be $m$ distinct items in the set $I$ of $n$ items.
Let $b \in \mathbb{R}^m_{> 0}$ be a vector where $b_i$ is the number of items of type $i$.
Therefore, $n = \sum_{i=1}^m b_i$. Let $N \defeq |\Ccal|$.
Let $A$ be an $m$-by-$N$ matrix where $A[i, C]$
is the number of items of type $i$ in \config{} $C$.
Then $A$ is called the \config{} matrix of $I$.
Let $\vecone \in \mathbb{R}^N$ be a vector whose each component is 1.

For every \config{} $C$, suppose we pack $x_C$ bins according to $C$.
Then the total number of bins used is $\vecone^Tx$.
The number of items of type $i$ that got packed is
$\sum_{C \in \Ccal} A[i, C] x_C = (Ax)_i$.
Therefore, the optimal solution to BP is given by
the optimal integral solution to $\covLP(A, b, \vecone)$.
$\covLP(A, b, \vecone)$ is called the \config{} LP of $I$
(also known as the Gilmore-Gomory LP of $I$).

Finding an approximately optimal (not necessarily integral) solution to
the \config{} LP of $I$ is also an important problem.
The algorithm of Karmarkar and Karp for classic-BP \cite{karmarkar-karp}
requires a $(1+1/n)$-approximate solution to $\covLP(A, b, \vecone)$.
The Round-and-Approx framework of Bansal and Khan \cite{bansal2014binpacking},
which is used to obtain the best-known approximation factor for 2GBP,
requires a $(1+\eps)$-approximate solution to $\covLP(A, b, \vecone)$.

\subsection{Solving the \Config{} LP using \texorpdfstring{$\covLPsolve$}{covLP-solve}}

\textbf{Indexing convention:}
Instead of using an integer $j \in [N]$ to index the columns in
the \config{} matrix $A$ and the entries in a feasible solution $x$,
we will index them by the corresponding \config{} $C$.
Hence, instead of writing $A[i, j]$ and $x_j$, we will write $A[i, C]$ and $x_C$.
Similarly, $\indexFindHyp$ will return a \config{} instead of an integer.

\begin{lemma}
\label{thm:config-lp-bounds}
Let $\covLP(A, b, \vecone)$ be the \config{} LP of a bin-packing instance $I$
having $n$ items. Then $1 \le \opt(\covLP(A, b, \vecone)) \le n$.
\end{lemma}
\begin{proof}
\Config{}s that contain only a single item are called singleton \config{}s.
Let $x_C = 0$ when $C$ is not a singleton \config{} and $x_C = b_i$ when $C$
is a singleton \config{} of item type $i$.
Then $x$ is a feasible solution to $\covLP(A, b, \vecone)$ and $\vecone^Tx = n$.
Therefore, $\opt(\covLP(A, b, \vecone)) \le n$.

Let $r^* = \opt(\covLP(A, b, \vecone))$.
Let $x^*$ be an optimal solution to $\covLP(A, b, \vecone)$.
Let $i$ be an arbitrary number in $[m]$
($m$ is the number of distinct items in $I$).
Since $x^*$ is feasible,
\[ b_i \le (Ax^*)_i = \sum_{C \in \mathcal{C}} A[i, C]x^*_C
\le \sum_{C \in \mathcal{C}} b_ix^*_C = b_ir^*
\implies 1 \le r^* \qedhere \]
\end{proof}

To solve $\covLP(A, b, \vecone)$ using $\covLPsolve$,
we need to compute $q$, an upper-bound on\\$\opt(\covLP(A, b, \vecone))$,
and $\rho$, an upper-bound on $q\max_{i=1}^m \max_{C \in \Ccal} A[i,C]/(b_i\vecone_C)$.
By \cref{thm:config-lp-bounds}, we can select $q \defeq n$.
Since $A[i,C] \le b_i$, we can choose $\rho \defeq n$.

The cost oracle simply outputs 1 for every input.
Let $a_C$ be the column of $A$ corresponding to \config{} $C$.
Then $a_C \in \mathbb{R}^m_{\ge 0}$ and the $i\Th$ coordinate of $a_C$
is the number of items of type $i$ in \config{} $C$.
Therefore, for any \config{} $C$, can get $a_C$ in $O(m)$ time.

For any \config{} $C$, define the function
$D_C: \mathbb{R}^m_{\ge 0} \mapsto \mathbb{R}_{\ge 0}$ as
\[ D_C(y) \defeq y^TAe_C = \sum_{i=1}^m y_iA[i,C] \]
Then for $\eta \in (0, 1]$, an $\eta$-weak index-finding oracle for $I$
is an algorithm that takes as input $y \in \mathbb{R}^m_{\ge 0}$
and returns $\widehat{C} \in \Ccal$ such that
$D_{\widehat{C}}(y) \ge \eta \max_{C \in \Ccal} D_C(y)$.
Note that if we assign profit $y_i$ to items of type $i$,
then $D_C(y)$ is the profit of \config{} $C$.
Therefore, an $\eta$-weak index-finding oracle
is an $\eta$-approximation algorithm for KS.

Now that we have the oracles ready, we can call
$\covLPsolve(I, n, n, \eps, \eta)$ to get a
$(1+\eps+\eps^2)/\eta$-approximate solution to $\covLP(A, b, \vecone)$.
Let us now look at the time complexity of this solution.

\begin{theorem}
\label{thm:bplp-time}
Let $I$ be a set of $n$ items, of which there are $m$ distinct items.
Assume we have an $\eta$-approximate algorithm for KS
that runs in time $O(T(m, n))$, for some function $T$ where $T(m, n) \ge m$.
Then $\covLPsolve(I, n, n, \eps, \eta)$ runs in time $O(MUT(m, n))$.
Here $M \in O(\log(n/(\eps\eta)))$ and $U \in \Otild(mn/(\eta\eps^3))$
(as defined in \cref{thm:covlps-time}).
\end{theorem}
\begin{proof}
By \cref{thm:covlps-time},
the time taken by non-oracle operations is $O(MUm)$,
the time taken by the product oracle is $O(MUm)$,
and the time taken by calls to the algorithm for KS is $O(MUT(m, n))$.
\end{proof}

Therefore, if $T(m, n) \in \poly(n)$, then $\covLPsolve$ gives us a
polynomial-time algorithm for solving the \config{} LP of $I$.

Note that $\eta$ can be very small here, i.e., this algorithm works
even if the approximation factor of KS is very bad.
As far as we know, all previous algorithms for approximately solving the \config{} LP
of BP assumed a PTAS for KS.
For many variants of KS, no PTAS is known.

\section{Comparison with Prior Work}
\label{sec:prior-work}

Many algorithms for solving general and special LPs exist.
In this section, we will look at the algorithms that have been used in the past
to solve implicitly-defined covering LPs, especially the \config{} LP of some variants of BP,
and why they cannot be used for other variants of BP.

\subsection{Ellipsoid Algorithm}

The Ellipsoid algorithm by Khachiyan \cite{khachiyan-ellipsoid},
in addition to being the first polynomial-time algorithm for linear programming,
can solve LPs that are implicitly defined.
Specifically, it uses a \emph{separation oracle},
which takes a vector $x$ as input, and either claims that $x$ is feasible
or outputs a constraint of the LP that is violated by $x$.
This is useful for solving LPs where the number of constraints
is super-polynomial in the input size
(Gr\"otschel, Lov\'asz and Schrijver \cite{gls-ellipsoid}
give many examples of this).

Let us see how the ellipsoid algorithm may be used for solving
the \config{} LP of a bin-packing instance.
Let $\covLP(A, b, \vecone)$ be implicitly defined in terms of input $I$.
$\covLP(A, b, \vecone)$ has $m$ constraints, where $m \le |I|$,
but the number of variables, $N$, can be super-polynomial.
We therefore compute the dual $D$ of $\covLP(A, b, \vecone)$,
that has $m$ variables and $N$ constraints.
We will solve $D$ using the Ellipsoid algorithm and then use that solution of $D$
to obtain a solution to $\covLP(A, b, c)$.
This is what $D$ looks like:
\[ \max_{y \in \mathbb{R}^m} b^Ty \textrm{ where } A^Ty \le \vecone \textrm{ and } y \ge 0 \]
The separation oracle for $D$ takes a vector $y$ as input
and checks if $A^Ty \le \vecone$.
\[ A^Ty \le \vecone
\iff \max_{C \in \Ccal} (A^Ty)_C \le 1
\iff \max_{C \in \Ccal} \sum_{i=1}^m y_iA[i, C] \le 1 \]
If we interpret $y_i$ as the profit of item $i$,
then $\sum_{i=1}^m y_iA[i, C]$ is the profit of \config{} $C$.
Therefore, the separation oracle is the decision version of the knapsack problem.
Specifically, the separation oracle should either claim that the optimal profit is at most 1,
or it should output a \config{} of profit more than 1.
Since the decision version of the knapsack problem is known to be NP-complete,
we cannot design a polynomial-time separation oracle.

Gr\"otschel, Lov\'asz and Schrijver \cite{gls-ellipsoid} gave a variant of
the Ellipsoid algorithm (which we will hereafter refer to as the GLS algorithm)
that can approximately solve an LP using an approximate separation oracle.
Karmarkar and Karp \cite{karmarkar-karp} modified the GLS algorithm to
solve the dual of the \config{} LP of classic-BP,
and described how to obtain a solution to the \config{} LP
using a solution to the dual.
Their algorithm, however, requires an FPTAS for classic-KS.
Our algorithm $\covLPsolve$ doesn't have such strict requirements,
and can work with very poorly-approximated algorithms for KS.

\subsection{Plotkin-Shmoys-Tardos Algorithm}

Plotkin, Shmoys and Tardos \cite{plotkin1995fast} gave algorithms for solving
the fractional covering problem (see \cref{defn:frac-cov}) and the fractional packing problem.
Our algorithm $\fracCoverHyp$ is obtained by slightly modifying their algorithm.
The following theorem is their most relevant result to us:
\begin{theorem}[Theorem 3.10 in \cite{plotkin1995fast}]
For $0 < \eps < 1$, given a $(1-\eps/2)$-weak point-finding oracle, the algorithm
of \cite{plotkin1995fast} $(1-\eps)$-weakly solves the fractional covering problem.
\end{theorem}
The above result holds only for a sufficiently small $\eps$.
\cite{plotkin1995fast} doesn't explicitly state how small $\eps$ should be,
but even the optimistic case of $\eps < 1$ tells us that for an $\eta$-weak $\pointFind$,
we require $\eta > 1/2$.
However, we are interested in the case where $\eta$ can be very small.

Moreover, there is a large gap between $(1-\eps/2)$ and $(1-\eps)$
when $\eps$ is large enough. For example, for $\eps = 1/3$,
$\pointFind$ is $5/6$-weak, but their algorithm will only give us a $2/3$-weak solution
to the fractional covering problem.
Our modified algorithm $\fracCover$, on the other hand, outputs a solution
that is roughly $5/6$-weak for this example.

We did not focus on optimizing the running time of our algorithm;
instead, we focused on getting as small an approximation factor as possible.
Our algorithm is, therefore, slower than that of \cite{plotkin1995fast}.

\section{The Fractional Covering Problem}
\label{sec:frac-cover}

Recall that in the problem $\fcovHyp(A, b, P)$, we need to find
$x \in P$ such that $Ax \ge b$ or claim that no such $x$ exists.
Also, $\rho \ge \widthHyp(A, b, P)$.

\subsection{Optimization Version of fcov}

Let us try to frame $\fcov(A, b, P)$ as an optimization problem.

\begin{definition}
\label{defn:ofcov}
For the problem $\fcov(A, b, P)$, let $\lambda(x) \defeq \max_{\lambda} (Ax \ge \lambda b)$.
The problem $\ofcov(A, b, P)$ is defined as
\[ \argmax_{x \in P} \lambda(x) \]
Let $x^*$ be the optimal solution to $\ofcov(A, b, P)$
and let $\lambda^* \defeq \lambda(x^*)$.
Then $x \in P$ is said to be $\eps$-optimal for $\ofcov(A, b, P)$
iff $\lambda(x) \ge (1-\eps)\lambda(x^*)$.
\end{definition}

\begin{claim}
\[ \lambda(x) = \min_{i=1}^m \frac{(Ax)_i}{b_i} \]
So $\lambda(x)$ can be computed using the \prodOrHyp.
\end{claim}

Note that $\fcov(A, b, P)$ is unsatisfiable iff $\lambda(x^*) < 1$,
and otherwise $x^*$ is a solution to $\fcov(A, b, P)$.
However, we can't directly use this fact to solve $\fcov(A, b, P)$,
since it may be very hard to compute $x^*$.
So instead, we'll compute an $\eps$-optimal solution $\xhat$ to $\ofcov(A, b, P)$.
Then $\lambda(\xhat)$ is an approximation to $\lambda(x^*)$,
since $(1-\eps)\lambda(x^*) \le \lambda(\xhat) \le \lambda(x^*)$.

\begin{claim}
\label{thm:opt-to-feas}
If $x$ is $\eps$-optimal for $\ofcov(A, b, P)$, then
\[ \lambda(x) < 1-\eps
\implies \lambda^* \le \frac{\lambda(x)}{1-\eps} < 1
\implies \fcov(A, b, P) \textrm{ has no solution } \]
\[ \lambda(x) \ge 1-\eps
\implies Ax \ge (1-\eps)b
\implies x \textrm{ is } (1-\eps)\textrm{-approx for } \fcov(A, b, P) \]
\end{claim}

We'll now focus on finding an $\eps$-optimal solution to $\ofcov(A, b, P)$.

\subsection{Weak Duality}

\begin{definition}
\label{defn:dfcov}
For the problem $\ofcov(A, b, P)$, define
\[ C(y) \defeq \max_{x \in P} y^T A x. \]
Define the problem $\dfcov(A, b, P)$ as
\[ \argmin_{y \in \mathbb{R}^m_{\ge 0}} \frac{C(y)}{y^Tb} \]
We call it the dual problem of $\ofcov$.
\end{definition}

\begin{lemma}[Weak duality]
\label{thm:weak-duality}
Let $\xhat \in P$. Then
\[ \lambda(\xhat)y^Tb \le y^TA\xhat \le C(y) \]
\end{lemma}
\begin{proof}
\begin{align*}
\lambda(\xhat)(y^Tb) &= y^T(\lambda(\xhat)b)
\\ &\le y^TA\xhat  \tag{$y \ge 0$ and $A\xhat \ge \lambda(\xhat)b$}
\\ &\le \max_{x \in P} y^TAx = C(y)
\qedhere \end{align*}
\end{proof}

\subsection{Relaxed Optimality Conditions}

If we could find $x \in P$ and $y \ge 0$ such that $\lambda(x)y^Tb \ge y^TAx \ge C(y)$,
then weak duality would imply that $x$ is optimal for $\ofcov$ and $y$ is optimal for $\dfcov$.
To find approximate optima, we slightly relax these conditions.
\[ \label{eqn:condI} (1+\eps_1)\lambda(x)y^Tb \ge y^TAx \tag{condition $\CondI(\eps_1)$} \]
\[ \label{eqn:condII} C(y) - y^TAx \le \eps_2 C(y) + \eps_3 \lambda(x) y^Tb
\tag{condition $\CondII(\eps_2, \eps_3)$} \]
Condition $\CondII(\eps_2, \eps_3)$ can equivalently be written as
\[ y^TAx \ge (1-\eps_2)C(y) - \eps_3 \lambda(x) y^Tb \]

\begin{lemma}
\label{thm:rlx-to-opt}
Suppose $x \in P$ and $y \ge 0$ satisfy conditions $\CondI(\eps_1)$
and $\CondII(\eps_2, \eps_3)$, where $0 < \eps_1, \eps_2, \eps_3 < 1$.
Let $\eps' \defeq (\eps_1 + \eps_2 + \eps_3)/(1 + \eps_2 + \eps_3)$.
Then $x$ is $\eps'$-optimal for $\ofcov(A, b, P)$.
\end{lemma}
\begin{proof}
By weak duality (\thmdepcref{thm:weak-duality}{}), we get
$\lambda^* \le C(y)/b^Ty$.

Conditions $\CondI(\eps_1)$ and $\CondII(\eps_2, \eps_3)$ give us
\begin{align*}
& (1+\eps_1)\lambda(x) y^Tb \ge y^TAx \ge (1-\eps_2)C(y) - \eps_3\lambda(x) y^Tb
\\ &\implies (1 + \eps_1 + \eps_3)\lambda(x) y^Tb \ge (1-\eps_2)C(y)
\\ &\implies \lambda(x) \ge \frac{1-\eps_2}{1 + \eps_1 + \eps_3} \lambda^*
    = (1-\eps')\lambda^*
\\ &\implies x \textrm{ is } \eps'\textrm{-optimal}
\qedhere \end{align*}
\end{proof}

We'll now focus on finding $x \in P$ and $y \ge 0$ such that
conditions $\CondI(\eps_1)$ and $\CondII(\eps_2, \eps_3)$ are satisfied.

\subsection{Dual Fitting}

\begin{definition}
For some $\alpha > 0$, define the vector $y_{\alpha}(x) \in \mathbb{R}^m_{> 0}$ as
\[ y_{\alpha}(x)_i \defeq \frac{1}{b_i}\exp\left( - \frac{\alpha(Ax)_i}{b_i} \right) \]
\end{definition}

\begin{lemma}
\label{thm:pot-bounds}
\[ \exp(-\alpha\lambda(x)) \le b^Ty_{\alpha}(x) \le m\exp(-\alpha\lambda(x)) \]
\end{lemma}
\begin{proof}
\begin{align*}
b^Ty_{\alpha}(x) &= \sum_{i=1}^m \exp\left(-\frac{\alpha(Ax)_i}{b_i}\right)
\\ &\in [1, m] \max_{i=1}^m \exp\left(-\frac{\alpha(Ax)_i}{b_i}\right)
\\ &= [1, m] \exp\left(-\alpha \min_{i=1}^m \frac{(Ax)_i}{b_i} \right)
= [1, m] \exp(-\alpha \lambda(x))
\qedhere \end{align*}
\end{proof}

\begin{lemma}
\label{thm:c1-sat}
Let $0 < \eps < 1$ and let
\[ \beta \defeq \frac{2}{\eps}\ln\left(\frac{4m}{\eps}\right) \]
(note that $\beta \ge 1$). Then
\[ \alpha\lambda(x) \ge \beta \implies (x, y_{\alpha}(x)) \textrm{ satisfies $\CondIHyp(\eps)$} \]
\end{lemma}
\begin{proof} See lemma 3.2 in \cite{plotkin1995fast}. \end{proof}

\begin{lemma}
\label{thm:width-is-positive}
$\width(A, b, P) = 0 \implies \fcov(A, b, P)$ is unsatisfiable.
\end{lemma}
\begin{proof}
We know that $\width(A, b, P) \ge 0$. Suppose $\width(A, b, P) = 0$.
Then $\forall x \in P, Ax = 0$, so $Ax \ge b$ cannot be true.
\end{proof}

\Cref{thm:width-is-positive} means that if $\rho = 0$,
then we know that $\fcov(A, b, P)$ is unsatisfiable.
So we'll assume from now on that $\rho > 0$.

\begin{theorem}
\label{thm:pot-dec}
Let $\epsSigma, \epsAlpha, \eps_2, \eps_3 \in (0, 1)$ be constants.
Suppose $x \in P$, $\xtild \in P$, $\alpha \in \mathbb{R}_{> 0}$,
$\eta \in (0, 1]$ and $\sigma \in \mathbb{R}_{> 0}$
satisfy the following properties:
\begin{itemize}
\item $\lambda(x) > 0$.
\item $(x, y_{\alpha}(x))$ doesn't satisfy condition $\CondIIHyp(\eps_2, \eps_3)$.
Denote $y_{\alpha}(x)$ by $y$ for simplicity.

\item $y^TA\xtild \ge \eta \hyperref[defn:dfcov]{C}(y)$.

\item $\sigma \le \epsSigma/(\alpha\rho)$.
Let $\xhat = (1-\sigma)x + \sigma\xtild$.
Denote $y_{\alpha}(\xhat)$ by $\yhat$ for simplicity.

\item ${\displaystyle \alpha \ge \frac{2}{\lambda(x)\epsAlpha}
\ln\left(\frac{4m}{\epsAlpha}\right)}$.

\item ${\displaystyle \eta \ge (1-\eps_2)\left(\frac{1+\epsSigma}{1-\epsSigma}\right)}$.
\end{itemize}

Then $\xhat \in P$ and
\[ \frac{b^Ty - b^T\yhat}{b^Ty}
> \alpha\sigma\lambda(x)\eps_3\left(1-\epsSigma\right) \frac{\eta}{1-\eps_2} \]
\end{theorem}
Intuitively, this means that under certain conditions, a point $x \in P$ can be
replaced by a different point in $P$ such that $b^Ty_{\alpha}(x)$ decreases by a significant factor.
The reasons for choosing such conditions will be apparent on reading
\cref{sec:frac-cover:improve-cover}.

\begin{proof}
\begin{align*}
\sigma &\le \frac{\epsSigma}{\alpha\rho}
\le \frac{\epsSigma}{\rho} \frac{\lambda(x)\epsAlpha}{2\ln\left(4m/\epsAlpha\right)}
\tag{$\alpha \ge \frac{2}{\lambda(x)\epsAlpha}\ln\left(\frac{4m}{\epsAlpha}\right)$}
\\ &\le \frac{\epsSigma\epsAlpha}{2\ln\left(4m/\epsAlpha\right)}
\tag{$\lambda(x) \le \rho$}
\\ &< \frac{1}{2\ln(4m)}
\le \frac{1}{2}
\tag{$\epsSigma, \epsAlpha < 1$}
\end{align*}
Therefore, $\sigma \le 1$, so $\xhat$ is a convex combination of $x$ and $\xtild$.
Therefore, $\xhat \in P$.

Let $\lambda_i \defeq (Ax)_i/b_i$, $\lamtild_i \defeq (A\xtild)_i/b_i$
and $\lamhat_i \defeq (A\xhat)_i/b_i$. Then
\begin{align*}
\xhat - x &= \sigma(\xtild - x)
& \lamhat_i &= (1-\sigma)\lambda_i + \sigma\lamtild_i
& \lamhat_i - \lambda_i &= \sigma(\lamtild_i - \lambda_i)
\\ y_i &= \frac{1}{b_i}e^{-\alpha\lambda_i}
& \yhat_i &= \frac{1}{b_i}e^{-\alpha\lamhat_i}
\end{align*}
\[ \frac{\yhat_i}{y_i}
= \exp(-\alpha(\lamhat_i - \lambda_i))
= \exp(-\alpha\sigma(\lamtild_i - \lambda_i)) \]

Let $\delta \defeq -\alpha\sigma(\lamtild_i - \lambda_i)$.
By definition of $\rho$ and $\widthHyp$, $\lambda_i \le \rho$ and $\lamtild_i \le \rho$.
Therefore, $\smallabs{\lamtild_i - \lambda_i} \le \rho$.
This gives us
$\abs{\delta} = \alpha\sigma\smallabs{\lamtild_i - \lambda_i}
\le \alpha\sigma\rho \le \epsSigma < 1$.

For $\abs{\delta} \le 1$, we get $e^{\delta} < 1 + \delta + \delta^2
\le 1 + \delta + \abs{\delta}\epsSigma$.

\begin{align*}
b_i\yhat_i &= b_iy_ie^{\delta}
\\ &< b_iy_i (1 + \delta + \epsSigma\abs{\delta})
\\ &= b_iy_i - \alpha\sigma y_ib_i(\lamtild_i - \lambda_i)
+ \alpha\sigma\epsSigma\smallabs{y_ib_i(\lamtild_i - \lambda_i)}
\\ &= b_iy_i - \alpha\sigma (y_i(A\xtild)_i - y_i(Ax)_i)
+ \alpha\sigma\epsSigma\abs{y_i(A\xtild)_i - y_i(Ax)_i}
\end{align*}
\begin{align*}
\implies \frac{b_iy_i - b_i\yhat_i}{\alpha\sigma}
&> (y_i(A\xtild)_i - y_i(Ax)_i) - \epsSigma\abs{y_i(A\xtild)_i - y_i(Ax)_i}
\\ &\ge (y_i(A\xtild)_i - y_i(Ax)_i) - \epsSigma(y_i(A\xtild)_i + y_i(Ax)_i)
\\ &= (1-\epsSigma)y_i(A\xtild)_i - (1+\epsSigma)y_i(Ax)_i
\end{align*}
\begin{align*}
\implies \frac{b^Ty - b^T\yhat}{\alpha\sigma}
&> (1-\epsSigma)y^TA\xtild - (1+\epsSigma)y^TAx
\\ &\ge (1-\epsSigma)\eta C(y) - (1+\epsSigma)y^TAx
\end{align*}
Since condition $\CondII(\eps_2, \eps_3)$ is not satisfied, we get
\[ (1-\eps_2)C(y) > y^TAx + \eps_3\lambda(x) b^Ty \]
\begin{align*}
\implies & \frac{b^Ty - b^T\yhat}{\alpha\sigma}
> (1-\epsSigma)\eta C(y) - (1+\epsSigma)y^TAx
\\ &> (1-\epsSigma)\eta \frac{y^TAx + \eps_3\lambda(x) b^Ty}{1-\eps_2}
- (1+\epsSigma)y^TAx
\\ &= \left((1-\epsSigma)\frac{\eta}{1-\eps_2} - (1 + \epsSigma)\right)y^TAx
+ (1-\epsSigma)\frac{\eta\eps_3}{1-\eps_2} \lambda(x) b^Ty
\end{align*}
\[ \eta \ge (1-\eps_2)\frac{1+\epsSigma}{1-\epsSigma}
\implies (1-\epsSigma)\frac{\eta}{1-\eps_2} - (1 + \epsSigma) \ge 0 \]
Therefore,
\begin{align*}
& \frac{b^Ty - b^T\yhat}{\alpha\sigma}
    > (1-\epsSigma) \frac{\eta\eps_3}{1-\eps_2} \lambda(x) b^Ty
\\ &\implies \frac{b^Ty - b^T\yhat}{b^Ty}
\ge \alpha\sigma\lambda(x)\eps_3(1-\epsSigma)\frac{\eta}{1-\eps_2}
\qedhere \end{align*}
\end{proof}

\subsection{Algorithm \texorpdfstring{$\improveCover$}{improve-cover}}
\label{sec:frac-cover:improve-cover}

We'll now start building an algorithm for $\fcov(A, b, P)$,
where $A$, $b$, $P$ are implicitly defined by input $I$.

\begin{algorithm}[H]
\caption{$\improveCover(I, x, \epsSigma, \eps_1, \eps_2, \eps_3, \rho)$:\\
Requires $x \in P$, $\rho > 0$, $\lambda(x) > 0$, $0 < \epsSigma, \eps_1, \eps_2, \eps_3 < 1$ and
${\displaystyle (1-\eps_2)\frac{1+\epsSigma}{1-\epsSigma} \le \eta \le 1}$.}
\label{algo:improve-cover}
\begin{algorithmic}[1]
\State $\lambda_0 = \lambda(x)
\quad {\displaystyle \alpha \defeq \frac{4}{\lambda_0\eps_1}\ln\left(\frac{4m}{\eps_1}\right)}
\quad {\displaystyle \sigma \defeq \frac{\epsSigma}{\alpha\rho}}$
\While{$\lambda(x) \le 2\lambda_0$ and $(x, y_{\alpha}(x))$ doesn't satisfy
        condition $\CondIIHyp(\eps_2, \eps_3)$}
    \State $\xtild = \pointFindHyp(y_{\alpha}(x))$
        \Comment{now $y_{\alpha}(x)^TA\xtild \ge \eta C(y_{\alpha}(x))$}.
    \State \label{alg-line:improve-cover:x-upd}$x = (1-\sigma)x + \sigma\xtild$
\EndWhile
\State \texttt{success} = \texttt{true} if
    $(x, y_{\alpha}(x))$ satisfies condition $\CondIIHyp(\eps_2, \eps_3)$ else \texttt{false}
\State \Return ($x$, \texttt{success})
\end{algorithmic}
\end{algorithm}

\begin{lemma}
\label{thm:improve-cover-invs}
Let $x^{[0]}$ be the initial value of $x$.
Throughout $\improveCover$, the following conditions are satisfied:
\begin{itemize}
\item $x \in P$
\item $b^Ty_{\alpha}(x) \le b^Ty_{\alpha}(x^{[0]})$
\item ${\displaystyle \lambda(x) \ge \frac{3}{4}\lambda_0}$
\item ${\displaystyle \alpha \ge \frac{3}{\lambda(x)\eps_1}\ln\left(\frac{4m}{\eps_1}\right)}$
\end{itemize}
\end{lemma}
\begin{proof}
These conditions are satisfied at the beginning of the algorithm.

Assume that the conditions are satisfied at the beginning of the \texttt{while} loop body.
Then the conditions for \thmdepcref{thm:pot-dec}{} are satisfied, so
$\xhat = (1-\sigma)x + \sigma\xtild \in P$
and $b^Ty_{\alpha}(\xhat) \le b^Ty_{\alpha}(x) \le b^Ty_{\alpha}(x^{[0]})$.

Let $\lamhat = \lambda(\xhat)$.
By \thmdepcref{thm:pot-bounds}{}, we get
\begin{align*}
& e^{-\alpha\lamhat} \le b^Ty_{\alpha}(\xhat)
\le b^Ty_{\alpha}(x^{[0]}) \le me^{-\alpha\lambda_0}
\\ &\implies me^{\alpha\lamhat} \ge e^{\alpha\lambda_0}
\\ &\implies \lambda_0 - \lamhat \le \frac{\ln m}{\alpha}
= \ln m \frac{\lambda_0\eps_1}{4 \ln\left(\frac{4m}{\eps_1}\right)}
\le \frac{\lambda_0}{4}
\\ &\implies \frac{3\lambda_0}{4} \le \lamhat
\end{align*}
\[ \alpha = \frac{4}{\lambda_0\eps_1} \ln\left(\frac{4m}{\eps_1}\right)
\ge \frac{3}{\lamhat\eps_1} \ln\left(\frac{4m}{\eps_1}\right) \]
Therefore, the conditions are also satisfied after the \texttt{while} loop.
\end{proof}

\begin{corollary}
\label{thm:improve-cover-c1}
Throughout $\improveCover$, $(x, y_{\alpha}(x))$ satisfies condition $\CondIHyp(\eps_1)$.
\end{corollary}
\begin{proof}
Use ${\displaystyle \alpha
\ge \frac{3}{\lambda(x)\eps_1} \ln\left(\frac{4m}{\eps_1}\right)}$
from \thmdepcref{thm:improve-cover-invs}{}
and \thmdepcref{thm:c1-sat}{}.
\end{proof}

\begin{lemma}
\label{thm:improve-cover}
Let $\eps' \defeq (\eps_1 + \eps_2 + \eps_3)/(1 + \eps_1 + \eps_3)$.
If $\improveCover$ terminates and \texttt{success} is \texttt{true},
then $x$ is \hyperref[defn:ofcov]{$\eps'$-optimal} for $\ofcov(A, b, P)$.
If $\improveCover$ terminates and \texttt{success} is \texttt{false},
then $\lambda(x) > 2\lambda_0$.
\end{lemma}
\begin{proof}
By \thmdepcref{thm:improve-cover-c1}{},
$x$ satisfies condition $\CondI(\eps_1)$.
If \texttt{success} is \texttt{true}, then $x$ satisfies condition $\CondIIHyp(\eps_2, \eps_3)$.
So by \thmdepcref{thm:rlx-to-opt}{},
$x$ is $\eps'$-optimal for $\ofcov(A, b, P)$.

If \texttt{success} is \texttt{false}, then $x$ doesn't satisfy
condition $\CondIIHyp(\eps_2, \eps_3)$.
Since the \texttt{while} loop ended, $\lambda(x) > 2\lambda_0$.
\end{proof}

\begin{claim}
\label{thm:ln-lb}
$\forall x \in \mathbb{R}_{> 0}, \ln x > (x-1)/x$.
\end{claim}

\begin{theorem}
\label{thm:improve-cover-iters}
Let $T$ be the number of times $\improveCover$ calls $\pointFindHyp$.
Then $T$ is finite and
\[ T \le \ceil{\frac{4\rho}{3\epsSigma\eps_3\lambda_0}
\left(\ln m + \frac{4}{\eps_1}\ln\left(\frac{4m}{\eps_1}\right)\right)}. \]
\end{theorem}
\begin{proof}
Let $x^{[t]}$ be the value of $x$ after $t$ runs of the \texttt{while} loop.
Let $\psi(t) \defeq b^Ty_{\alpha}(x^{[t]})$.
So $\lambda(x^{[t]}) \le 2\lambda_0$ for all $t < T$.
By \thmdepcref{thm:pot-bounds,thm:improve-cover-invs}{}, we get
\begin{align*}
& e^{-\alpha(2\lambda_0)} \le e^{-\alpha\lambda(x^{[t]})}
\le b^Ty_{\alpha}(x^{[t]}) \le b^Ty_{\alpha}(x^{[0]}) \le me^{-\alpha\lambda_0}
\\ &\implies \frac{\psi(0)}{\psi(t)}
\le \frac{me^{-\alpha\lambda_0}}{e^{-\alpha(2\lambda_0)}} = me^{\alpha\lambda_0}
\end{align*}
Define
\[ \beta \defeq \frac{3}{4}\alpha\sigma\lambda_0\eps_3(1-\epsSigma)\frac{\eta}{1-\eps_2} \]
By \thmdepcref{thm:improve-cover-invs}{},
the conditions for \cref{thm:pot-dec} are satisfied for $x = x^{[j]}$ for all $j < T$.
By \thmdepcref{thm:pot-dec}{},
\[ \forall j < T, \frac{\psi(j+1)}{\psi(j)} < 1 - \beta
\implies \frac{\psi(0)}{\psi(t)} > (1-\beta)^{-t} \]
Therefore,
\[ (1-\beta)^{-t} < \frac{\psi(0)}{\psi(t)} \le me^{\alpha\lambda_0}
\implies t < \frac{\ln m + \alpha\lambda_0}{\ln(1/(1-\beta))}
\iff t \le \ceil{\frac{\ln m + \alpha\lambda_0}{\ln(1/(1-\beta))}} - 1 \]
Since this is true for all $t < T$, we get that $T$ is finite and
\[ T \le \ceil{\frac{\ln m + \alpha\lambda_0}{\ln(1/(1-\beta))}} \]
By \thmdepcref{thm:ln-lb}{}, $\ln((1-\beta)^{-1}) \ge \beta$.
\begin{align*}
\beta &= \frac{3}{4}\alpha\sigma\lambda_0
\left(1-\epsSigma\right)\frac{\eta\eps_3}{1-\eps_2}
= \frac{3\epsSigma\lambda_0}{4\rho}
\left(1-\epsSigma\right)\frac{\eta\eps_3}{1-\eps_2}
\tag{since $\sigma \defeq \epsSigma/(\alpha\rho)$}
\\ &\ge \frac{3\eps_3\epsSigma\lambda_0}{4\rho}
\left( 1 + \epsSigma\right)
> \frac{3\eps_3\epsSigma\lambda_0}{4\rho}
\tag{${\displaystyle \eta \ge (1-\eps_2)\frac{1+\epsSigma}{1-\epsSigma}}$}
\end{align*}
Therefore,
\begin{align*}
T \le \ceil{\frac{\ln m + \alpha\lambda_0}{\ln((1-\beta)^{-1})}}
&\le \ceil{\frac{1}{\beta}\left(\ln m + \frac{4}{\eps_1}\ln\left(\frac{4m}{\eps_1}\right)\right)}
\\ &\le \ceil{\frac{4\rho}{3\eps_3\epsSigma\lambda_0}
\left(\ln m + \frac{4}{\eps_1}\ln\left(\frac{4m}{\eps_1}\right)\right)}
\qedhere \end{align*}
\end{proof}

\subsection{Algorithm \texorpdfstring{$\fracCover$}{frac-cover}}

\subsubsection{Starting with Good \texorpdfstring{$\lambda(x)$}{lambda(x)}}

The number of iterations within $\improveCover$ depends inversely on $\lambda_0$.
Therefore, we have to ensure that $\lambda_0$ isn't too small.
Let $e_i$ be the $i\Th$ standard basis vector for $\mathbb{R}^m$.
So $e_i^TAx = (Ax)_i$.

\begin{lemma}
\label{thm:get-seed}
Let $x^{(i)} \defeq \pointFindHyp(e_i)$.
Let $\ddot{x} \defeq \frac{1}{m}\sum_{i=1}^m x^{(i)}$.
If $(Ax^{(i)})_i < \eta b_i$ for some $i \in [m]$, then $\fcov(A, b, P)$ is unsatisfiable.
Otherwise, $\lambda(\ddot{x}) \ge \eta/m$.
\end{lemma}
\begin{proof}
Since $\pointFind$ is $\eta$-weak,
$(Ax^{(i)})_i \ge \eta \max_{x \in P} (Ax)_i$.
\begin{align*}
(Ax^{(i)})_i < \eta b_i
&\implies \max_{x \in P} (Ax)_i < b_i
\implies \forall x \in P, Ax \not\ge b
\\ &\implies \fcov(A, b, P) \textrm{ is unsatisfiable.}
\end{align*}
Suppose $(Ax^{(i)})_i \ge \eta b_i$ for all $i \in [m]$.
Since each $x^{(i)} \in P$, we get $\ddot{x} \in P$. Also,
\[ (A\ddot{x})_i = \frac{1}{m} \sum_{j=1}^m (Ax^{(j)})_i
\ge \frac{1}{m} (Ax^{(i)})_i \ge \frac{\eta}{m} b_i
\implies \lambda(\ddot{x}) \ge \frac{\eta}{m} \qedhere \]
\end{proof}

\Cref{thm:get-seed} gives us \cref{algo:get-seed} ($\getSeed$).

\begin{algorithm}[H]
\caption{$\getSeed(I, \eta)$}
\label{algo:get-seed}
\begin{algorithmic}[1]
\For{$i \in [m]$}
    \State $x^{(i)} = \pointFindHyp(e_i)$
        \Comment{so $e_i^TAx^{(i)} \ge \eta \max_{x \in P} e_i^TAx$.}
    \If{$(Ax^{(i)})_i < \eta b_i$}
        \State \Return \texttt{null}
    \EndIf
\EndFor
\State \Return ${\displaystyle \ddot{x} \defeq \frac{1}{m}\sum_{i=1}^m x^{(i)}}$
\end{algorithmic}
\end{algorithm}

\subsubsection{Iteratively Running \texorpdfstring{$\improveCover$}{improve-cover}}

\begin{algorithm}[H]
\caption{$\fracCover(I, \rho, \eps, \eta)$: Returns either $x \in P$ or \texttt{null}.
\\ Requires $\rho \ge \widthHyp(A, b, P)$ and $\eps, \eta \in (0, 1]$.}
\label{algo:frac-cover}
\begin{algorithmic}[1]
\If{$\rho \texttt{ == } 0$}
    \State \label{alg-line:frac-cover:rho}\Return \texttt{null}
\EndIf
\State $x = \getSeedHyp(I, \eta)$
\If{$x \texttt{ == } \texttt{null}$}
    \State \label{alg-line:frac-cover:ret-null-1}\Return \texttt{null}
\EndIf
\State Let $\epsSigma \defeq \eps/(6+5\eps)$ and $\eps_1 \defeq \eps_3 \defeq \eps/3$.
\State Let ${\displaystyle \eps_2 \defeq 1 - \eta\frac{1-\epsSigma}{1+\epsSigma}}$
and ${\displaystyle \eps' \defeq \frac{\eps_1 + \eps_2 + \eps_3}{1 + \eps_1 + \eps_3}}$.
\While{\texttt{true}}
    \If{$\lambda(x) \ge 1$}\label{alg-line:frac-cover:lambda-check-1}
        \State \label{alg-line:frac-cover:ret-exact}\Return $x$
    \EndIf
    \State $(x, \texttt{success})
        = \improveCoverHyp(x, \epsSigma, \eps_1, \eps_2, \eps_3, \rho)$.
    \If{\texttt{success}}
        \Comment{$x$ is $\eps'$-optimal.}
        \State \label{alg-line:frac-cover:ret-e-opt}%
\Return $x \textrm{ if } \lambda(x) \ge 1-\eps' \textrm{ else } \texttt{null}$
    \EndIf
\EndWhile
\end{algorithmic}
\end{algorithm}

\begin{lemma}
\label{thm:epsdash}
In $\fracCover$, $1 - \eps' = \eta/(1+\eps)$.
\end{lemma}
\begin{proof}
\[ 1-\eps' = \frac{1-\eps_2}{1 + \eps_1 + \eps_3}
= \eta \frac{\gamma - \epsSigma}{\gamma + \epsSigma}
\frac{1}{1 + \eps_1 + \eps_3}
= \eta \frac{1-\eps/(6+5\eps)}{1+\eps/(6+5\eps)}\frac{1}{1 + 2\eps/3}
= \frac{\eta}{1 + \eps} \qedhere \]
\end{proof}

\restate{\cref{thm:frac-cover}}{\thmFracCover}
\begin{proof}
When $\fracCover$ returns \texttt{null} at \cref{alg-line:frac-cover:rho},
then $\fcov(A, b, P)$ is unsatisfiable by \thmdepcref{thm:width-is-positive}{thm:frac-cover}.

$\getSeed$ and $\improveCover$ have output in $P$.
Therefore, throughout the execution of $\fracCover$, $x \in P$.

When $\fracCover$ returns \texttt{null} at \cref{alg-line:frac-cover:ret-null-1},
it does so because $\getSeed$ returned \texttt{null}.
Then as per \thmdepcref{thm:get-seed}{thm:frac-cover}, $\fcov(A, b, P)$ is unsatisfiable.

When $\fracCover$ returns $x$ at \cref{alg-line:frac-cover:ret-exact},
then $\lambda(x) \ge 1 \implies Ax \ge b$.

$\fracCover$ returns at \cref{alg-line:frac-cover:ret-e-opt} when \texttt{success} is \texttt{true}.
Then by \thmdepcref{thm:improve-cover}{thm:frac-cover}, $x$ would be $\eps'$-optimal.
By \thmdepcref{thm:opt-to-feas}{thm:frac-cover},
we get that if $\fracCover$ returns \texttt{null},
then $\fcov(A, b, P)$ is unsatisfiable.
Otherwise, $x$ is a $(1-\eps')$-approximate solution to $Ax \ge b$,
i.e. $Ax \ge (1-\eps')b$.
By \thmdepcref{thm:epsdash}{thm:frac-cover}, $Ax \ge (\eta/(1+\eps))b$.
\end{proof}

\begin{lemma}
\label{thm:fc-pf-calls}
Let $\fcov(A, b, P)$ be implicitly defined by input $I$,
where $A \in \mathbb{R}^{m \times N}_{\ge 0}$.
Suppose $\fracCover(I, \rho, \eps, \eta)$ calls $\pointFindHyp$ $T$ times. Then
\[ T \le U \defeq \Uexpr \in \Otild\left(\frac{m\rho}{\eta\eps^3}\right) \]
\end{lemma}
\begin{proof}
$\getSeed$ calls $\pointFind$ $m$ times.
The other $T-m$ calls to $\pointFind$ occur inside $\improveCover$.

For $j \ge 0$, let $x^{[j]}$ be the value of $x$ given as input to
the $(j+1)\Th$ run of $\improveCover$.
Suppose $\improveCover$ is called at least $t$ times.
Then $\lambda(x^{[t-1]}) < 1$ (see \cref{alg-line:frac-cover:lambda-check-1}).
Also, \texttt{success} was \texttt{false} for the first $t-1$ calls to $\improveCover$,
so by \thmdepcref{thm:improve-cover}{},
$\forall j \in [t-1], \lambda(x^{[j]}) > 2\lambda(x^{[j-1]})$.
Therefore, $2^{t-1}\lambda(x^{[0]}) \le \lambda(x^{[t-1]}) < 1$.

By \thmdepcref{thm:get-seed}{},
$\lambda(x^{[0]}) \ge \eta/m$. Therefore,
\begin{align*}
& 2^{t-1} < \frac{m}{\eta}
\implies t-1 < \lg\left(\frac{m}{\eta}\right)
\iff t \le \ceil{\lg\left(\frac{m}{\eta}\right)}
\end{align*}
Therefore, $\improveCover$ is called at most $\ceil{\lg(m/\eta)}$ times.

For each input $x$ to $\improveCover$, $\lambda(x) \ge \eta/m$.
So, by \thmdepcref{thm:improve-cover-iters}{},
the number of times $\pointFind$ is called in each call to $\improveCover$ is at most
\begin{align*}
&\ceil{\frac{4\rho}{3\epsSigma\eps_3\lambda_0}
\left(\ln m + \frac{4}{\eps_1}\ln\left(\frac{4m}{\eps_1}\right)\right)}
\le \ceil{\frac{4m\rho}{3\eta\epsSigma\eps_3}
\left(\ln m + \frac{4}{\eps_1}\ln\left(\frac{4m}{\eps_1}\right)\right)}
\\ &= \ceil{\frac{4m\rho}{3\eta}\frac{6+5\eps}{\eps}\frac{3}{\eps}
\left(\ln m + \frac{12}{\eps}\ln\left(\frac{12m}{\eps}\right)\right)}
\le \ceil{\frac{312m\rho(1+\eps)}{\eta\eps^3}\ln\left(\frac{12m}{\eps}\right)}
\end{align*}
This gives us
\[ T \le \Uexpr \qedhere \]
\end{proof}

\begin{lemma}
\label{thm:fc-time-2}
Let $\fcov(A, b, P)$ be implicitly defined in terms of input $I$,
where $A \in \mathbb{R}^{m \times N}_{\ge 0}$.
Let $\tau$ be an upper-bound on the support of the output of $\pointFindHyp$.
Suppose $\fracCover(I, \rho, \eps, \eta)$ calls $\pointFind$ $T$ times.
Then,
\begin{itemize}
\item $\fracCover$ makes at most $T$ calls to the \prodOrHyp{}.
For every input $x$ to the product oracle, $|\support(x)| \le \tau$.
\item The running time of $\fracCover$, excluding the time taken by oracles,
is $O(T(m + \tau))$.
\item The solution $\xhat$ returned by $\fracCover$ has $|\support(\xhat)| \le T\tau$.
\end{itemize}
\end{lemma}
\begin{proof}
$\getSeed$ calls $\pointFind$ $m$ times.
The other $T-m$ calls to $\pointFind$ occur inside $\improveCover$.
Let $x^{[0]}$ be the output of $\getSeed$.
Assume $x^{[0]}$ is not null (since otherwise the proof is trivial).
For $t \in [T-m]$, in the $t\Th$ run of
line \ref{alg-line:improve-cover:x-upd} in $\improveCover$,
let $\sigma_t$ be the value of variable $\sigma$,
let $\xtild^{[t]}$ be the output of $\pointFind$,
and let $x^{[t]}$ be the new value of variable $x$.
Therefore, $x^{[t]} = (1-\sigma_t)x^{[t-1]} + \sigma_t\xtild^{[t]}$,
and the output of $\fracCover$ is $x^{[T-m]}$.

In $\getSeed$, the product oracle is called $m$ times
--- to compute $(Ax^{(i)})_i$ for $i \in [m]$.
The inputs to these calls have a support of size at most $\tau$.
We can take the mean of these product oracle outputs to get $Ax^{[0]}$.

After $\getSeed$, the product oracle is never called directly --- it is only needed
to compute $\lambda(x^{[t]})$ and $y_{\alpha}(x^{[t]})$
and to check condition $\CondIIHyp(\eps_2, \eps_3)$ for all $0 \le t \le T-m$.
To compute $Ax^{[t]}$, we won't call the product oracle on $x^{[t]}$,
since the support of $x^{[t]}$ can be too large.
Instead, we'll compute it indirectly from $x^{[t-1]}$ and $\xtild^{[t]}$ using
$Ax^{[t]} = (1-\sigma_t)(Ax^{[t-1]}) + \sigma_t(A\xtild^{[t]})$.
Note that we already know $Ax^{[0]}$ and $|\support(\xtild^{[t]})| \le \tau$.
Therefore, we need to call the product oracle $T$ times in $\fracCover$,
and each input to the product oracle has support of size at most $\tau$.

$|\support(x^{[0]})| \le m\tau$, since $x^{[0]}$ is the mean of $m$ outputs of $\pointFind$.
Each modification of $x$ increases $|\support(x)|$ by at most $\tau$.
Therefore, $|\support(x^{[T-m]})| \le T\tau$.

A crucial observation for reducing the running time of the algorithm is that
we don't really need to keep track of intermediate values of $x$;
we only need the final value $x^{[T-m]}$.
This is because $x$ isn't used directly anywhere in the algorithm.
We only need it indirectly in the form of $\lambda(x)$ and $y_{\alpha}(x)$
and for checking if $(x, y_{\alpha}(x))$ satisfies condition $\CondIIHyp(\eps_2, \eps_3)$.
But for these purposes, knowing $Ax$ is enough.
Actually computing $x^{[t]}$ for all $t$, as line \ref{alg-line:improve-cover:x-upd}
in $\improveCover$ suggests, can be wasteful and costly,
since $x^{[t]}$ can have a large support and $T$ can be large.
We will, therefore, compute $x^{[T-m]}$ directly
without first computing the intermediate values of $x$.

On solving the recurrence $x^{[t]} = (1-\sigma_t)x^{[t-1]} + \sigma_t\xtild^{[t]}$, we get
\[ x^{[T-m]} = \gamma_{T-m}\left(x^{[0]}
+ \sum_{t=1}^{T-m} \frac{\sigma_t\xtild^{[t]}}{\gamma_t}\right) \]
where $\gamma_t \defeq \prod_{j=1}^t (1-\sigma_j)$.
This way, $x^{[T-m]}$ can be computed in $O(T\tau)$ time.

Other than oracle calls and computing $x^{[T-m]}$,
the time taken by $\fracCover$ is $O(Tm)$:
$O(m)$ in $\getSeed$ and $O(m)$ in each iteration of
the \texttt{while} loop in $\improveCover$.
\end{proof}

\begin{proof}[Proof of \cref{thm:frac-cover-time}]
Follows from \thmdepcref{thm:fc-pf-calls,thm:fc-time-2}{thm:frac-cover-time}.
\end{proof}

\section{Covering LPs}
\label{sec:covLP-solve}

Recall that $\covLP(A, b, c)$ is this linear program:
\[ \min_{x \in \mathbb{R}^N} c^Tx \textrm{ where } Ax \ge b \textrm{ and } x \ge 0 \]
where $A \in \mathbb{R}_{\ge 0}^{m \times N}$, $b \in \mathbb{R}^m_{>0}$
and $c \in \mathbb{R}^N_{\ge 0}$.
$A, b, c$ are defined implicitly by an input $I$.

Before we try to approximately solve $\covLP(A, b, c)$,
let us see why an optimal solution always exists.
\begin{lemma}
A covering LP is feasible and has bounded objective value.
\end{lemma}
\begin{proof}
Consider $\covLP(A, b, c)$. Let $a_i^T$ be the $i\Th$ row of $A$.
For any feasible solution $x$, we have $(Ax)_i = a_i^Tx \ge b_i > 0$.
Hence, $a_i \neq 0$. Let
\[ \xhat \defeq \sum_{i=1}^m \frac{b_i}{\norm{a_i}^2} a_i > 0. \]

For all $i, j \in [m]$, we have $a_i^Ta_j \ge 0$. This gives us
\[ (A\xhat)_i = a_i^T\xhat
= a_i^T\left(\sum_{j=1}^m \frac{b_j}{\norm{a_j}^2} a_j\right)
\ge a_i^T \left(\frac{b_i}{\norm{a_i}^2} a_i\right) = b_i \]
Therefore, $\xhat$ is feasible for $\covLP(A, b, c)$.

For any feasible solution $x$, we have $c^Tx \ge 0$ because $c > 0$ and $x \ge 0$.
Therefore, $\covLP(A, b, c)$ is bounded.
\end{proof}

We will try to solve $\covLP(A, b, c)$ by binary searching on the objective value $c^Tx$.

Given $r \in \mathbb{R}_{\ge 0}$, we want to either
find a feasible solution to $\covLP(A, b, c)$ of objective value $r$,
or claim that no solution exists of objective value $r$.
This is equivalent to $\fcovHyp(A, b, P_r)$,
where $P_r \defeq \{x: c^Tx = r \textrm{ and } x \ge 0\}$.

Let $r^* = \opt(\covLP(A, b, c))$. Then $\fcov(A, b, P_r)$ has a solution iff $r \ge r^*$.
If we could exactly solve $\fcov$, then finding a $(1+\eps)$-approximate solution
to $\covLP(A, b, c)$ is straightforward:
use binary search to find $r$ such that $r^* \le r \le (1+\eps)r^*$,
and then solve $\fcov(A, b, P_r)$ to get a feasible solution $x$ such that $c^Tx = r$.
However, we can't do this because we can't solve $\fcov(A, b, P_r)$ exactly.
Nevertheless, a similar approach can be used to approximately solve $\covLP(A, b, c)$
if we can weakly solve $\fcov(A, b, P_r)$.

\subsection{Solving the Fractional Covering Problem}

$\covLPsolve$ receives the following inputs:
\begin{itemize}
\item $I$: the input used to implicitly define $\covLP(A, b, c)$.
\item $q$: an upper-bound on $\opt(\covLP(A, b, c))$.
\item $\rho$: an upper-bound on
${\displaystyle q\max_{i=1}^m \max_{j=1}^N \frac{A[i,j]}{b_ic_j}}$.
\item $\eps \in (0, 1]$.
\item $\eta \in (0, 1]$.
\end{itemize}
$\covLPsolve$ is also provided a \columnOrHyp{} for $A$, a \costOrHyp{} for $c$,
and an $\eta$-weak \indexOrHyp{}.
We will now design an algorithm $\fracCovII(I, r, \rho, \eps, \eta)$
that weakly solves $\fcov(A, b, P_r)$ for any $r \in [0, q]$
using these inputs and oracles.

As per \cref{thm:frac-cover}, $\fracCover$ can $\eta/(1+\eps)$-weakly solve
$\fcov(A, b, P_r)$ if we can provide it the following values and oracles:
\begin{itemize}
\item an upper-bound on $\width(A, b, P_r)$.
\item $\eps, \eta \in (0, 1]$.
\item A \prodOrHyp{}.
\item An $\eta$-weak \pointOrHyp{}.
\end{itemize}
Providing $\eta$ and $\eps$ is trivial, since they are inputs to $\covLPsolve$.
We will now prove that $\rho$ is an upper-bound on $\width(A, b, P_r)$
and see how to implement $\pointFind$ using $\indexFind$ and the cost oracle,
and how to implement the product oracle using the column oracle.

\subsubsection{Upper-Bounding Width}

\begin{lemma}
\label{thm:pr-extreme-points}
Let $P_r \defeq \{x: c^Tx = r \textrm{ and } x \ge 0\}$.
Then $P_r$ is bounded, and the extreme points of $P_r$ are
\[ S \defeq \left\{ \frac{r}{c_j} e_j: j \in [N] \right\} \]
Here $e_j$ is the $j\Th$ standard basis vector of $\mathbb{R}^N$.
\end{lemma}
\begin{proof}
Let $j \in [N]$. Then $c^Tx = r \implies c_jx_j \le r \implies x_j \in [0, r/c_j]$.
Therefore, $P_r$ is bounded.

Let $x \in P_r$. Then
\[ x = \sum_{j=1}^N \left(\frac{x_jc_j}{r}\right)\left(\frac{re_j}{c_j}\right) \]
Therefore, all points in $P_r$ can be represented as convex combinations of $S$.
Hence, the set of extreme points of $P_r$ is a subset of $S$.
Since $S$ is a basis of $\mathbb{R}^N$, no point in $S$
can be represented as a linear combination of the other points in $S$.
Therefore, $S$ is the set of extreme points of $P_r$.
\end{proof}

\begin{claim}
\label{thm:lin-opt-extreme}
Let $h: \mathbb{R}^N \mapsto \mathbb{R}$ be a linear function and
$P \subseteq \mathbb{R}^N$ be a polytope.
Then there exists an extreme point $\xhat$ of $P$ such that
$h(\xhat) = \max_{x \in P} h(x)$.
\end{claim}

\begin{lemma}
\label{thm:pr-width}
Let $P_r \defeq \{x: c^Tx = r \textrm{ and } x \ge 0\}$. Then
\[ \width(A, b, P_r) = r \max_{i=1}^m \max_{j=1}^N \frac{A[i, j]}{b_ic_j} \]
\end{lemma}
\begin{proof}
Let $a_i^T$ be the $i\Th$ row of $A$.
Since a linear function $(a_i^Tx/b_i)$ over polytope $P_r$ is maximized at its extreme points
(see \thmdepcref{thm:lin-opt-extreme,thm:pr-extreme-points}{}), we get
\begin{align*}
\width(A, b, P_r) &= \max_{i=1}^m \max_{x \in P_r} \frac{a_i^Tx}{b_i}
= \max_{i=1}^m \max_{j=1}^N \frac{1}{b_i}a_i^T\left(\frac{r}{c_j}e_j\right)
= r \max_{i=1}^m \max_{j=1}^N \frac{A[i, j]}{b_ic_j}
\qedhere \end{align*}
\end{proof}

Since $r \le q$, we get that $\rho \ge \width(A, b, P_r)$.

\subsubsection{Implementing the Point-Finding Oracle}

\begin{lemma}
\label{thm:point-vs-index}
Let $P_r \defeq \{x: c^Tx = r \textrm{ and } x \ge 0\}$. Then
\[ \max_{x \in P_r} y^TAx = r\max_{j=1}^N y^TA\pfrac{e_j}{c_j} \]
\end{lemma}
\begin{proof}[Proof sketch]
Use \thmdepcref{thm:lin-opt-extreme}{} with $h(x) = y^TAx$
and use \thmdepcref{thm:pr-extreme-points}{}.
\end{proof}
\begin{corollary}
\label{thm:point-using-index}
Let $\indexFind$ be an $\eta$-weak for $\fcov(A, b, c)$.
If we set $\pointFind(y) = (r/c_k)e_k$, where $k \defeq \indexFind(y)$,
then $\pointFind$ is $\eta$-weak for $\fcov(A, b, P_r)$,
where $P_r \defeq \{x: c^Tx = r \textrm{ and } x \ge 0\}$.
Here $c_k$ is obtained as $\operatorname{cost-oracle}(k)$.
\end{corollary}
\begin{proof}
Let $\xtild = \pointFind(y)$.
Then by \thmdepcref{thm:point-vs-index}{}, we get
\[ y^TA\xtild = r y^TA\pfrac{e_k}{c_k} \ge r\eta \max_{j=1}^N y^TA\pfrac{e_j}{c_j}
= \eta \max_{x \in P_r} y^TAx \qedhere \]
\end{proof}

\subsubsection{Implementing the Product Oracle}

Let $a_j$ be the $j\Th$ column of $A$. To compute $Ax$, simply use
\[ Ax = \sum_{j=1}^N x_ja_j \]
Therefore, we can implement the product oracle over $A$
using $|\support(x)|$ calls to the column oracle.

\subsubsection{Summary}

The description of $\fracCovII$ is now complete, and we get the following result:
\begin{theorem}
\label{thm:frac-cov-2}
Let $\covLP(A, b, c)$ be defined implicitly in terms of $I$.
Let $q$ be an upper-bound on $\opt(\covLP(A, b, c))$ and
$\rho$ be an upper-bound on $q\max_{i=1}^m \max_{j=1}^N A[i,j]/(b_ic_j)$.
Let $P_r \defeq \{x: c^Tx = r \textrm{ and } x \ge 0 \}$.
Let $\indexFindHyp$ be an $\eta$-weak index-finding oracle.

Then we can implement an $\eta$-weak $\pointFindHyp$ for $\fcov(A, b, P_r)$
using a single call to $\indexFind$ and the \costOrHyp{}.
For the \prodOrHyp{}, we can compute $Ax$ using $|\support(x)|$
calls to the \columnOrHyp{}.
Also, for every $\xtild$ output by $\pointFind$, $|\support(\xtild)| \le 1$.

Furthermore, $\fracCovII(I, r, \rho, \eps, \eta)$ will $\eta/(1+\eps)$-weakly solve
$\fcovHyp(A, b, P_r)$, and\\$\fracCovII(I, r, \rho, \eps, \eta)$ works by
returning the output of $\fracCoverHyp((I, r), \rho, \eps, \eta)$.
\end{theorem}
\thmdep{thm:pr-width,thm:point-using-index,thm:frac-cover}{}

\subsection{Algorithm Based on Binary Search}

Let $r^* \defeq \opt(\covLP(A, b, c))$ and $\mu \defeq \eta/(1+\eps)$.
Note that $r^* > 0$, since every feasible solution is non-zero,
and hence has positive objective value.

\begin{algorithm}[H]
\caption{$\covLPsolve(I, q, \rho, \eps, \eta)$:
\\ Finds a $(1+\delta)/\mu$-approximate solution to $\covLP(A, b, c)$,
where $A, b, c$ are implicitly defined by $I$.
Here $\eps, \eta \in (0, 1]$,
$q$ is an upper-bound on $r^* \defeq \opt(\covLP(A, b, c))$,
$\rho$ is an upper-bound on $q\max_{i=1}^m \max_{j=1}^N A[i,j]/(b_ic_j)$,
$P_r \defeq \{x: c^Tx = r \textrm{ and } x \ge 0 \}$,
and $\indexFindHyp$ is an $\eta$-weak index-finding oracle.}
\label{algo:cov-lp-solve}
\begin{algorithmic}[1]
\State $\delta \defeq \eps^2/(1+\eps)$
\State $\alpha = 0$
\State $\beta = q$
\State $\xhat = \fracCovII(I, q, \rho, \eps, \eta)$
\While{$\beta > (1+\delta)\alpha$}
    \State $r = (\alpha+\beta)/2$
    \State $\yhat = \fracCovII(I, r, \rho, \eps, \eta)$
    \If{$\yhat$ is \texttt{null}}
        \State $\alpha = r$
    \Else
        \State $\beta = r$
        \State $\xhat = \yhat$
    \EndIf
\EndWhile
\State \Return $(\alpha, \beta, \xhat)$
\end{algorithmic}
\end{algorithm}

\begin{definition}
Let $g: [0, q] \mapsto \{0, 1\}$ be a function where
\\ $g(r) = 0$ iff $\fracCovII(I, r, \rho, \eps, \eta)$ returns \texttt{null}.
\end{definition}
Every call to $\fracCovII$ in $\covLPsolve$ probes a point $r$ in the interval $[0, q]$
and gives us $g(r)$.

\begin{lemma}
\label{thm:covlps-part}
When $r < \mu r^*$, $g(r)$ is always 0. When $r \ge r^*$, $g(r)$ is always 1.
(When $\mu r^* \le r < r^*$, $g(r)$ may be 0 or 1.)
\end{lemma}
\begin{proof}
Let $r < \mu r^*$. Assume $g(r) = 1$. This means that $\fracCovII$ returned
a $\mu$-approximate solution $x$ to $\fcov(A, b, P_r)$
(by \thmdepcref{thm:frac-cov-2}{}), i.e.,
$Ax \ge \mu b$ and $c^Tx = r$.
Therefore, $x/\mu$ is feasible for $\covLP(A, b, c)$ and $c^T(x/\mu) = r/\mu < r^*$.
This is a contradiction, since we found a feasible solution to $\covLP(A, b, c)$
of objective value less than the optimum.
Therefore, $g(r) = 0$.

Let $r \ge r^*$ and let $x^*$ be an optimal solution to $\covLP(A, b, c)$.
Therefore, $Ax^* \ge b$ and $c^Tx^* = r^*$.
Let $x \defeq (r/r^*)x^*$. Then $Ax \ge (r/r^*)b \ge b$ and $c^Tx = r$.
Therefore, $x$ is a feasible solution to $\fcov(A, b, P_r)$.
This means $\fcov(A, b, P_r)$ is satisfiable, so $\fracCovII$ cannot return \texttt{null}.
Therefore, $g(r) = 1$.
\end{proof}
\begin{lemma}
\label{thm:covlps-invs}
Throughout $\covLPsolve$'s execution,
$g(\alpha) = 0$, $g(\beta) = 1$, $c^T\xhat = \beta$ and $A\xhat \ge \mu b$.
(assuming the \texttt{while} loop body is executed atomically).
\end{lemma}
\begin{proof}
When $\yhat = \texttt{null}$, then $g(r) = 0$ and $\alpha = r$.
Otherwise, $g(r) = 1$ and $\beta = r$.

$c^T\xhat = \beta$ and $A\xhat \ge \mu b$
follow from $\xhat \neq \texttt{null}$,
$\xhat = \fracCovII(I, \beta, \rho, \eps, \eta)$
and the fact that $\fracCovII$ can $\mu$-weakly solve $\fcov(A, b, P_r)$
(by \thmdepcref{thm:frac-cov-2}{}).
\end{proof}

\restate{\cref{thm:cov-lp-solve}}{\thmCovLPSolve}
\begin{proof}
Let $(\alpha, \beta, \xhat) = \covLPsolve(I, q, \rho, \eps, \eta)$.
By \thmdepcref{thm:covlps-invs}{thm:cov-lp-solve},
$g(\alpha) = 0$, $g(\beta) = 1$, $c^T\xhat = \beta$ and $A\xhat \ge \mu b$.
By \thmdepcref{thm:covlps-part}{thm:cov-lp-solve}, $\alpha < r^*$ and $\beta \ge \mu r^*$.
Since the algorithm terminated, $\beta \le (1+\delta)\alpha$.
This gives us $r^* \le c^T(\xhat/\mu) \le r^*(1+\delta)/\mu$
and $A(\xhat/\mu) \ge b$.
Therefore, $\xhat$ is a $(1+\delta)/\mu$-approximate solution to $\covLP(A, b, c)$.
\[ \frac{1+\delta}{\mu} = \left(1 + \frac{\eps^2}{1+\eps}\right) \frac{1+\eps}{\eta}
= \frac{1+\eps+\eps^2}{\eta} \qedhere \]
\end{proof}

\begin{theorem}
\label{thm:covlps-iters}
Suppose the \texttt{while} loop in $\covLPsolve$ runs $T$ times. Then
\[ T \le 2 + \lg\pfrac{q}{r^*} + \lg\pfrac{1}{\eta} + 2\lg\left(\frac{1}{\eps} + 1\right) \]
\end{theorem}
\begin{proof}
Let $\alpha_t$ and $\beta_t$ denote the values of
$\alpha$ and $\beta$ after $t$ runs of the \texttt{while} loop.
Then $\alpha_0 = 0$ and $\beta_0 = q$.

Suppose $\alpha_t = 0$ for all $t < p$. Then $\beta_{p-1} = q/2^{p-1}$.
$g(\beta_{p-1}) = 1$ by \thmdepcref{thm:covlps-invs}{},
and $\beta_{p-1} \ge \mu r^*$ by \thmdepcref{thm:covlps-part}{}.
Therefore, $p \le \lg(q/(\mu r^*)) + 1$.
Let $p$ be the largest possible such that $\alpha_t = 0$ for all $t < p$.
Then $\alpha_p = q/2^p$ and $\beta_p = \beta_{p-1} = q/2^{p-1}$.

$\beta_t - \alpha_t$ halves in each iteration. So for any $t \ge p$,
$\beta_t - \alpha_t = (\beta_p - \alpha_p)/2^{t-p}$.
Let $t = p + \ceil{\lg(1/\delta)}$.
Assume that $T > t$. Since the \texttt{while} loop ran the $(t+1)\Th$ time,
$\beta_t > (1+\delta)\alpha_t$.
\[ \frac{\beta_t - \alpha_t}{\alpha_t}
= \frac{\beta_p - \alpha_p}{2^{\ceil{\lg(1/\delta)}}\alpha_t}
\le \frac{q/2^p}{\alpha_p/\delta} = \delta
\implies \beta_t \le (1+\delta)\alpha_t \]
This is a contradiction. Therefore,
\begin{align*}
T &\le p + \ceil{\lg\pfrac{1}{\delta}}
\le 2 + \lg\pfrac{q}{r^*} + \lg\pfrac{1}{\mu\delta}
\\ &= 2 + \lg\pfrac{q}{r^*} + \lg\pfrac{1}{\eta} + 2\lg\left(\frac{1}{\eps} + 1\right)
\qedhere \end{align*}
\end{proof}

\restate{\cref{thm:covlps-time}}{\thmCovLPSolveTime}
\begin{proof}[Proof sketch]
By \thmdepcref{thm:covlps-iters}{thm:covlps-time},
$\covLPsolve$ calls $\fracCovII$ at most $M$ times.
By \thmdepcref{thm:frac-cov-2}{thm:covlps-time}, $\tau = 1$
and every call to $\fracCovII$ results in one call to $\fracCover$,
and every call to $\pointFind$ results in one call to $\indexFind$.
The rest follows from \thmdepcref{thm:frac-cover-time}{thm:covlps-time}.
\end{proof}

\section{Future Work}
\label{sec:future-work}

$\fracCover$ is based on a simplified version of the Plotkin-Shmoys-Tardos
algorithm \cite{plotkin1995fast} for fractional covering.
We did not focus on optimizing the running time of $\fracCover$;
instead, we focused on getting as small an approximation factor as possible,
even when the point-finding oracle could be very weak.
\cite{plotkin1995fast} uses many tricks to get a low running time.
We didn't adapt those tricks to our algorithm,
so our algorithm is not as fast as theirs.
For example, they use different values of $\epsilon$
for each call to $\improveCover$,
and they have fast randomized versions of their algorithms.

The most important of these tricks, in our opinion, is the one
that reduces the dependence on $\rho$.
The number of times our algorithm $\fracCover$ calls the point-finding oracle
varies linearly with $\rho$. But for some applications,
$\rho$ can be super-polynomial in the input size.
Section 4 of \cite{plotkin1995fast} explains a possible approach to fix this.
The number of times their algorithm calls the point-finding oracle
is linear in $\log\rho$.

Another direction of work would be to adapt our techniques to the
fractional packing problem. \cite{plotkin1995fast} already have an algorithm for this,
but their algorithm doesn't work for very small values of $\eta$
(when using a packing analogue of the $\eta$-weak point-finding oracle).

\end{document}